\documentclass[oneside,english,utf8]{amsart}
\usepackage[T1]{fontenc}
\usepackage[latin9]{inputenc}
\usepackage{float}
\usepackage{amsbsy}
\usepackage{amstext}
\usepackage{amsthm}
\usepackage{amssymb}
\PassOptionsToPackage{normalem}{ulem}
\usepackage{ulem}

\makeatletter

\floatstyle{ruled}
\newfloat{algorithm}{tbp}{loa}
\providecommand{\algorithmname}{Algorithm}
\floatname{algorithm}{\protect\algorithmname}

\numberwithin{equation}{section}
\numberwithin{figure}{section}
\theoremstyle{plain}
\newtheorem{thm}{\protect\theoremname}
  \theoremstyle{definition}
  \newtheorem{example}[thm]{\protect\examplename}
  \theoremstyle{definition}
  \newtheorem{defn}[thm]{\protect\definitionname}
  \theoremstyle{remark}
  \newtheorem{claim}[thm]{\protect\claimname}
  \theoremstyle{plain}
  \newtheorem{lem}[thm]{\protect\lemmaname}
  \theoremstyle{plain}
  \newtheorem{cor}[thm]{\protect\corollaryname}

\newcommand\SP[1][1.5em]{%
  \mbox{\kern.06em\vrule height.3ex}%
  \vbox{\hrule width#1}%
  \hbox{\vrule height.3ex}}

\makeatother

\usepackage{babel}
  \providecommand{\claimname}{Claim}
  \providecommand{\corollaryname}{Corollary}
  \providecommand{\definitionname}{Definition}
  \providecommand{\examplename}{Example}
  \providecommand{\lemmaname}{Lemma}
\providecommand{\theoremname}{Theorem}

\begin{document}

\title{Data Reduction in markov model using em algorithm}

\author{Atanu Kumar Ghosh, Arnab Chakraborty}
\begin{abstract}
This paper describes a data reduction technique in case of a markov
chain of specified order. Instead of observing all the transitions
in a markov chain we record only a few of them and treat the remaining
part as missing. The decision about which transitions to be filtered
is taken before the observation process starts. Based on the filtered
chain we try to estimate the parameters of the markov model using
EM algorithm. In the first half of the paper we characterize a class
of filtering mechanism for which all the parameters remain identifiable.
In the later half we explain methods of estimation and testing about
the transition probabilities of the markov chain based on the filtered
data. The methods are first developed assuming a simple markov model
with each probability of transition positive, but then generalized
for models with structural zeroes in the transition probability matrix.
Further extension is also done for multiple markov chains. The performance
of the developed method of estimation is studied using simulated data
along with a real life data.
\end{abstract}


\maketitle

\section{Introduction}

Applications of statistical models often encounter with datasets which
are too large to store or to gain meaningful interpretation. This
motivates the need for a suitable data filtering mechanism which stores
only a subset of the available information such that based on the
observed data reasonable inferences can be drawn about the parameter.
This paper describes such a filtering mechanism in case of discrete
markov models. Discrete markov chains are the simplest dependent structure
that one can think of and are very useful for modeling a wide range
of scientific problems in nature. Some important applications include
modeling of dry and wet spells (P. J. Avery and D. A. Henderson (1999)),
deoxyribonucleic acid (DNA) sequences (P. J. Avery and D. A. Henderson
(1999) ), study of chronic diseases ( B. A. Craig and A. A. Sendi
(2002) ). 

Any stochastic process $\mathbf{X}=\{X_{1},X_{2},\cdots,X_{n}\}$
having a finite set $\mathcal{S}$ as its state space, is said to
be a markov process of order $s$ if 
\[
P(X_{n}=a_{n}|\,X_{n-1}=a_{n-1},X_{n-2}=a_{n-2},\cdots,X_{1}=a_{1})
\]
\[
=P(X_{n}=a_{n}|\,X_{n-1}=a_{n-1},X_{n-2}=a_{n-2},\cdots,X_{n-s}=a_{n-s})
\]
For notational convenience let us denote the state space as $\mathcal{S=}\{1,2,\cdots,k\}.$
Further we assume that the markov process has stationary transition
probabilities, which means 
\[
P(X_{n}=a_{n}|\,X_{n-1}=a_{n-1},X_{n-2}=a_{n-2},\cdots,X_{n-s}=a_{n-s})=p_{a_{n-s},\cdots,a_{n-1}:a_{n}}
\]
does not depend on $n.$ For $s=1$ we have a simple markov chain
with finite state space. Any markov chain of $s^{th}$ order can be
treated as a simple markov chain with suitable parameters. So in this
paper we will develop the methods assuming a simple markov chain which
will be equally applicable for any markov process of higher order.
A markov chain can be completely described by the initial state and
the set of transition probabilities. Here we shall consider the initial
state of a markov chain to be known and try to make inferences about
the transition probabilities based on the observed data. More specifically
inferences regarding the transition probability matrix can help us
to answer many specific questions regarding the markov process which
we usually encounter.

There is an extensive literature available on the statistical inferences
of finite markov chains based on complete data. Billingsley(4) gives
a good account of the mathematical aspects of different techniques
regarding inferences about the transition probabilities which includes
Whittle's formula, maximum likelihood and chi-square methods. Estimation
of transition probabilities and testing goodness of fit from a single
realisation of a markov chain has been studied by Bartlett(3). Goodman
and Anderson(1) derived the estimates of the transition probabilities
and their hypothesis when there are more than one realisation of a
single markov chain. Their paper also described the asymptotic properties
of the methods when the number of realisations increase. All these
works assume the observed data to be one or more long, unbroken observations
of the chain. In this paper we assume that there is a single realisation
of the markov chain which is not completely observed. The observed
broken chain which results from the filtering mechanism is therefore
not markov. 

Based on the filtering mechanism we will observe only certain transitions
of a markov chain and treat the remaining part of the chain as missing.
From the observed data we will estimate the transition probabilities
using EM algorithm. EM algorithm is a standard tool for maximum likelihood
estimation in incomplete data problems. Since the missingness in the
data occurs due to the filtering process, the data are not missing
at random (NMAR). Here the missing mechanism is nonignorable but known.
Unlike the conventional uses of EM algorithm where missingness occurs
naturally, here we introduce missingness deliberately to reduce the
size of the observed data. The E step of the EM algorithm requires
to find the conditional expectation of the missing data given the
observed data. This is achieved by defining the all possible missing
paths for a transition of any order and finding the probability of
the same. The standard error of the EM estimate is obtained by the
supplemented EM algorithm (SEM) (Meng and Rubin). Usually the standard
error of the EM estimate is obtained by inverting the observed information
matrix. In our case the observed likelihood cannot be obtained explicitly
and hence we avoid the calculation of the observed information matrix.
SEM is a technique to calculate asymptotic dispersion matrix of the
EM estimate without inverting the observed information matrix.

Section 2 describes the setup of the problem. Section 3 deals with
the identifiability issues of the parameter that arise due to filtering
of data. In section 2 we assume that the transition probability matrix
consists of all positive elements. This assumption is relaxed in section
3 where we allow some structural zeroes in the transition probability
matrix. We describe the additional modification we need in the filtering
mechanism due to such relaxation. Section 5 describes the methods
of estimation and testing the transition probabilities. In this section
we also describe the estimation of standard errors of the estimates
by the SEM algorithm. Section 6 describes the generalization of the
above methods in case of multiple markov chains. In section 7 we demonstrate
the methods developed using simulated data. A real life data analysis
is demonstrated in section 8. Section 9 is the appendix which has
the proofs of a major theorem of this paper.

\section{Setup}

Let $\mathbf{X}$ be a simple markov chain with finite state space
$\mathcal{S}=\{1,2,\cdots,k\}$ and transition probability matrix
$P=((p_{ij}))_{k\times k}.$ Let us first assume that $0<p_{ij}<1,\,\forall\,i,j$.
We shall relax this assumption later and consider the case where we
allow some $p_{ij}$'s to be zero. The transition probability matrix
$P$ satisfy the standard condition 
\[
P\mathbf{1=1}\;i.e.\quad{\displaystyle \sum_{j}p_{ij}=1\quad\forall i}.
\]
Hence we have $k^{2}-k$ independent parameters. We define the vector
of the parameters 
\[
\boldsymbol{\theta}=(p_{11},p_{12},\cdots,p_{1(k-1)},p_{21},p_{22},\cdots,p_{2(k-1)},\cdots,p_{k1},p_{k2},\cdots,p_{k(k-1)})^{'}
\]
\[
=(\theta_{1},\theta_{2},\cdots,\theta_{d})^{'}
\]
where $d=k^{2}-k$ and the parameter space is 
\[
\boldsymbol{\Theta}=\{\boldsymbol{\theta}:\sum_{j=1}^{k-1}p_{ij}<1,\quad\textrm{for }i=1,\cdots,k\}
\]
\[
=\{\boldsymbol{\theta}:\sum_{j=1}^{k-1}\theta_{j}<1\;,\sum_{j=0}^{k-1}\theta_{(i-1)k+j}<1,\quad\textrm{for }i=2,\cdots,k\}.
\]
Now suppose we consider a single realization $\mathbf{x}$ of the
chain and the number of transitions from state $i$ to state $j$
in this realization is $n_{ij}$. We assume that the markov process
is continued sufficiently long enough so that the realization $\mathbf{x}$
contains each transition at least once, that is, $n_{ij}>0$ for all
$i$ and $j.$ The matrix of transition count is 
\[
N=\begin{bmatrix}n_{11} & n_{12} & \cdots & \cdots & n_{1k}\\
n_{21} & n_{22} & \cdots & \cdots & n_{2k}\\
 &  & \vdots\\
n_{\overline{k-1}1} & n_{\overline{k-1}2} & \cdots & \cdots & n_{\overline{k-1}k}\\
n_{k1} & n_{k2} & \cdots & \cdots & n_{kk}
\end{bmatrix}.
\]
In this paper we propose a data acquisition protocol which suggests
that instead of observing the entire realization $\mathbf{x},$ we
record only some of the transitions and treat the remaining part of
the chain as missing. The decision about which transitions we record
is described in the form of a filter matrix $\mathbf{F}=((f_{ij}))_{1\leq i,j\leq k}$
which contains 0 and 1 as elements. In particular we record the transition
from state $i$ to state $j$$\;\textrm{if }f_{ij}=1$. If $\mathbf{X}$
is the complete chain then let $\phi_{\mathbf{F}}(\mathbf{X})$ denote
the chain filtered using $\mathbf{F}.$
\begin{example}
\label{example1}Consider a three state markov chain $\mathbf{x}$
as 

\[
112312232123331121331
\]
Suppose we are given a filter matrix 
\[
\mathbf{F}=\begin{bmatrix}1 & 0 & 0\\
0 & 1 & 0\\
1 & 1 & 0
\end{bmatrix}.
\]
Then the transitions we record in $\phi_{\mathbf{F}}(\mathbf{x})$
are
\begin{equation}
\begin{array}{c}
1\rightarrow1\\
2\rightarrow2\\
3\rightarrow1\\
3\rightarrow2
\end{array}\label{eq:1}
\end{equation}
Then the filtered chain is 
\[
11\SP312232\SP\SP\SP\SP311\SP\SP\SP31
\]

\end{example}
In the filtered chain the missing states are indicated by the symbol
$"\SP"$ which we call ``blank''. The example shows that besides
the transitions (\ref{eq:1}) there may be some transitions which
are indirectly recorded in the filtered chain (such as $2\rightarrow3$
is recorded even if $f_{23}=0).$ Any transition $i\rightarrow j$
may be recorded indirectly in the filtered chain if there exist some
states $a$ and $b$ such that $f_{ai}=1$ and $f_{jb}=1.$ Thus all
the transitions in the filtered chain may be classified into one of
the three categories: 
\begin{itemize}
\item directly recorded ($f_{ij}=1$)
\item indirectly recorded ($f_{ij}=0$ but the transition occurs in the
filtered chain, e.g. $2\rightarrow3$ in Example 1)
\item unobserved ( $f_{ij}=0$ and the transition does not appear in the
filtered chain e.g. $3\rightarrow3$ in Example 1)
\end{itemize}

\section{Identifiability}

In this section we shall discuss about the identifiability of the
parameters based on the filtered chain. We note that our filtered
chain no longer possesses the markov property. While applying the
filtering mechanism, if we record only a very few transitions then
all the parameters of the markov chain may not be identifiable. For
example in a markov chain with state space $\{1,2,\cdots,10\}$ if
we record only the transition $1\rightarrow1$ then some parameters
, say $p_{55}$, are not identifiable. We need to study how much data
we can throw away, so that the problem still remains identifiable.
Thus our main aim, in this section, will be to identify a class of
filter matrices so that data generated by any filter matrix of that
class will retain the identifiability of the parameters. But first
we define what is meant by identifiability of parameter on the basis
of a random sample. 
\begin{defn}
Let $\mathbf{X}$ be a random sample from a distribution characterized
by the parameter $\boldsymbol{\theta}$ and $L(\boldsymbol{\theta},\mathbf{x})$
be the likelihood. Then the parameter $\boldsymbol{\theta}$ is said
to be identifiable on the basis of $\mathbf{X}$ if for any two distinct
values $\boldsymbol{\theta}_{0}$ and $\boldsymbol{\theta}_{1}$ in
the parameter space 

\[
L(\boldsymbol{\theta}_{1},\mathbf{x})\neq L(\boldsymbol{\theta}_{2},\mathbf{x})
\]

\end{defn}
\bigskip{}

Suppose $\mathbf{X}$ is a random sample drawn from a population characterized
by the parameter $\boldsymbol{\theta}$. Let $\mathbf{Y}=g(\mathbf{X})$
be function of $\mathbf{X}$. Given $\mathbf{X}$ we can always construct
$\mathbf{Y}$ through $g$. So if $\boldsymbol{\theta}$ is identifiable
on the basis of $\mathbf{Y}$, we can identify $\boldsymbol{\theta}$
also from $\mathbf{X}$. On the contrary, if $\boldsymbol{\theta}$
is unidentifiable on the basis of $\mathbf{X}$, then it is also unidentifiable
on the basis of $\mathbf{Y}.$ This is because, if we assume $\boldsymbol{\theta}$
to be identifiable on the basis of $\mathbf{Y},$ then given $\mathbf{X}$,
we can construct $\mathbf{Y}$ through $g$ and then $\boldsymbol{\theta}$
can be identified from $\mathbf{X}$, which is a contradiction. Thus
in general we have the following two results:
\begin{claim}
\label{claim3}a) If $\boldsymbol{\theta}$ is identifiable on the
basis of $\mathbf{Y}$, then $\boldsymbol{\theta}$ is also identifiable
on the basis of $\mathbf{X}$.

\qquad{}\qquad{}b) If $\boldsymbol{\theta}$ is unidentifiable on
the basis of $\mathbf{X}$, then $\boldsymbol{\theta}$ is also unidentifiable
on the basis of $\mathbf{Y}$.
\end{claim}
In the present situation to prove that the parameters are identifiable
it is enough to consider a observed sample $\mathbf{x}$ such that
$\exists t$ such that $P_{\boldsymbol{\theta}}(\phi_{\mathbf{F}}(\mathbf{x})=t)\neq P_{\boldsymbol{\theta}'}(\phi_{\mathbf{F}}(\mathbf{x})=t)$
and prove that any two different values of the parameter $\boldsymbol{\theta}$
will yield different values of the observed likelihood $L_{obs}(\boldsymbol{\theta},\mathbf{x})$. 

Let $\mathcal{F}$ be the class of all $k\times k$ filter matrices.
We call a filter matrix $\mathbf{F}\in\mathcal{F}$ identifiable if
$P$ is identifiable with respect to $\phi_{\mathbf{F}}(\mathbf{X})$.
Let $\mathcal{I}_{\theta}\subseteq\mathcal{F}$ be the set of all
$k\times k$ filter matrices for which the parameter $\theta$ is
identifiable. Then $\mathcal{I}={\displaystyle \cap}\mathcal{I}_{\theta}$
is the set of identifiable filter matrices.

With this notation, the general fact stated in claim \ref{claim3}
is also applicable for the data generated by the filter matrices.
\begin{lem}
\label{lemma4}For $\mathbf{H},\mathbf{F}\in\mathcal{F}$, let $\phi_{\mathbf{H}}=g\circ\phi_{\mathbf{M}}$
for some function $g(.)$. Then $\mathbf{H}\in\mathcal{I}$ implies
$\mathbf{M}\in\mathcal{I}$ and $\mathbf{M}\in\mathcal{F}-\mathcal{I}$
implies $\mathbf{H}\in\mathcal{F}-\mathcal{I}.$\end{lem}
\begin{example}
Let 
\[
\mathbf{H}=\begin{bmatrix}1 & 0 & 0\\
0 & 1 & 0\\
1 & 1 & 0
\end{bmatrix}\qquad\textrm{and}\qquad\mathbf{M}=\begin{bmatrix}1 & 1 & 0\\
0 & 1 & 0\\
1 & 1 & 0
\end{bmatrix}
\]
$\mathbf{M}$ is same as $\mathbf{H}$ except that for $\mathbf{M}$
we directly observe one more transition $2\rightarrow3$ than $\mathbf{H}$.
Then $\phi_{\mathbf{H}}=g\circ\phi_{\mathbf{M}}$ and hence if $\mathbf{H}\in\mathcal{I}$
then $\mathbf{M}\in\mathcal{I}$.
\end{example}
In general there are $2^{k^{2}}$possible filter matrices in $\mathcal{F}.$
Instead of searching over all possible filter matrices we shall start
with some definite structures of filter matrices which are identifiable.
The above discussion motivates us to extend the identifiability over
a larger class of matrices. This requires some ordering of the filter
matrices in terms of the data we store. 
\begin{defn}
For filter matrix $\mathbf{M}=((m_{ij}))\in\mathcal{F}$ and $\mathbf{H}=((h_{ij}))\in\mathcal{F}$
we say $\mathbf{M}\succeq\mathbf{H}$ if $\forall i,j\quad h_{ij}=1\Rightarrow m_{ij}=1$
and $\mathbf{M}\preceq\mathbf{H}$ if $\forall i,j\quad h_{ij}=0\Rightarrow m_{ij}=0.$
\end{defn}
\medskip{}

\begin{lem}
\label{lemma7}a) If $\mathbf{H}\in\mathcal{I}$ and $\mathbf{M}\succeq\mathbf{H}$
then $\mathbf{M}\in\mathcal{I}.$

\qquad{}\qquad{}b) If $\mathbf{H}\in\mathcal{F}-\mathcal{I}$ and
$\mathbf{M}\preceq\mathbf{H}$ then $\mathbf{M}\in\mathcal{F}-\mathcal{I}.$\end{lem}
\begin{proof}
a) $\mathbf{M}\succeq\mathbf{H}$ implies $\phi_{\mathbf{H}}=g(\phi_{\mathbf{M}})$
for some $g(.).$ Using Lemma \ref{lemma4} we get $\mathbf{H}\in\mathcal{I}$
implies $\mathbf{M}\in\mathcal{I}.$

\qquad{}b) $\mathbf{M}\preceq\mathbf{B}$ implies $\phi_{\mathbf{M}}=g(\phi_{\mathbf{H}})$
for some $g(.).$ Using Lemma \ref{lemma4} we get $\mathbf{H}\in\mathcal{F}-\mathcal{I}$
implies $\mathbf{M}\in\mathcal{F}-\mathcal{I}$.
\end{proof}
Thus if any filter matrix $\mathbf{M}$ is identifiable, then all
filter matrices which stores more data than $\mathbf{M}$ are also
identifiable. This fact is also true for any subclass of filter matrices. 
\begin{defn}
If $\mathcal{D}\subseteq\mathcal{F}$, then the closure of $\mathcal{D}$
is defined as 
\[
\bar{\mathcal{D}}=\{\mathbf{F}\in\mathcal{F}\;:\mathbf{F}\succeq\mathbf{D}\;\textrm{for some }\mathbf{D}\in\mathcal{D}\}.
\]
\end{defn}
\begin{lem}
\label{lemma9}If $\mathcal{D}\subseteq\mathcal{I}$ then $\mathcal{\bar{D}}\subseteq\mathcal{I}.$\end{lem}
\begin{proof}
Let $\mathbf{M}\in\mathcal{\bar{D}}$. Then $\mathbf{M}\succeq\mathbf{D}\;\textrm{for some }\mathbf{D}\in\mathcal{D}$. 

Since $\mathcal{D}\subseteq\mathcal{I}$, we get $\mathbf{D}\in\mathcal{I}$.
Then Lemma \ref{lemma7} implies $\mathbf{M}\in\mathcal{I}$. 

This implies $\mathcal{\bar{D}}\subseteq\mathcal{I}.$
\end{proof}
Thus given any class of identifiable filter matrices $\mathcal{D}$
we can always extend it to a larger subclass of identifiable filter
matrices.

Our observed chain is a sequence of states and blanks $(\SP)$. Given
any observed chain we want to find the condition under which the conditional
probability of a given segment of the observed chain given the initial
state in the segment is identifiable. 
\begin{defn}
For any finite sequence $\pi$ of states or blanks $(\SP)$ we define
\[
S_{\pi}=\textrm{set of all filtered segments where }\pi\textrm{ occurs in consecutive positions}.
\]

We note that if $\pi_{1}\subseteq\pi_{2}$ then $S_{\pi_{1}}\supseteq S_{\pi_{2}}$. \end{defn}
\begin{lem}
\label{lemma11}For any filter matrix $\mathbf{F}$, if $P(S_{\pi})>0$
then $p_{\pi}$ is identifiable where $\pi$ is a sequence of states
or blanks which starts and ends with states and $p_{\pi}$ is the
conditional probability of the sequence $\pi$ given the initial state
in $\pi$.\end{lem}
\begin{proof}
Let $\pi$ start with the state $\alpha$ and end with the state $\beta$. 

Let 

\[
\mathcal{A}=\textrm{subchains that ends with }\alpha.
\]
\[
\mathcal{B}=\textrm{subchains that ends with the sequence }\pi.
\]
Then $\mathcal{B}\subseteq\mathcal{A}.$ Also $P(S_{\pi})>0$ implies
$P(\mathcal{B})>0$ which implies $P(\mathcal{A})>0$. 

Also from Markov property we get that $P(\mathcal{B}|\mathcal{A})=p_{\pi}.$
Thus if $p_{\pi}$ changes $P(\mathcal{B}|\mathcal{A})$ changes.
Since the conditional probability of a class of subchains changes,
the joint distribution of the entire filtered chain must also change.
Hence two distinct values of $p_{\pi}$ will give two distinct values
of the observed likelihood.Thus $p_{\pi}$ is identifiable. \end{proof}
\begin{cor}
For any filter matrix $F$ the parameter $p_{ij}$ is identifiable
if $P(S_{ij})>0.$
\end{cor}
As mentioned before we want to start with filter matrices of definite
structures which are identifiable and extend them to relatively larger
classes. With this view in mind, we define three classes of filter
matrices each of which will be sufficient for a filter matrix to be
identifiable. 

\renewcommand{\labelenumi}{\alph{enumi})}
\begin{description}
\item [{Class1}] We define $\mathbb{C}_{1}\subseteq\mathcal{F}$ which
consists of all filter matrix $\mathbf{F}=((f_{ij}))\:,\:1\leq i,j\leq k$
such that \end{description}
\begin{enumerate}
\item $\exists$ $\alpha$ such that $f_{\alpha j}=0,j=1,2,...,k$ i.e.
the $\alpha^{th}$ row of $\mathbf{F}$ is zero. 
\item $\exists$ $\beta$ such that $f_{i\beta}=0,i=1,2,...,k$ i.e. the
$\beta^{th}$column of $\mathbf{F}$ is zero.
\item $f_{pj}=1\textrm{ for exactly one }j,\:1\leq j\leq k,\qquad p=1,2,...,k\,,\,p\neq\alpha$
,i.e. except $\alpha^{th}$ row every other row must have exactly
one element $1$.
\item $f_{iq}=1\textrm{ for exactly one }i,\:1\leq i\leq k,\qquad q=1,2,...,k\,,\,q\neq\beta$
,i.e. except $\beta^{th}$ column every other column must have exactly
one element $1$.
\end{enumerate}
\bigskip{}

\begin{description}
\item [{Class2}] We define $\mathbb{C}_{2}\subseteq\mathcal{F}$ which
consists of all filter matrix $\mathbf{F}=((f_{ij}))\:,\:1\leq i,j\leq k$
such that\end{description}
\begin{enumerate}
\item $\exists$ $\alpha$ and $\beta$ such that $f_{i\alpha}=0,i=1,2,...,k$
and $f_{i\beta}=0,i=1,2,...,k$ i.e. the $\alpha^{th}$ and $\beta^{th}$
column of $F$ is zero.
\item $f_{iq}=1\textrm{ for at exactly one }i,\:1\leq i\leq k,\qquad q=1,2,...,k\,,\,q\neq\alpha,\beta$
,i.e. except $\alpha^{th}$ and $\beta^{th}$ column every other column
have exactly one element $1$.
\item $f_{\alpha j}=f_{\beta j}=1\quad1\leq j\leq k,\qquad,\,j\neq\alpha,\beta$,
i.e. except $\alpha^{th}$ and $\beta^{th}$ column every other element
of $\alpha^{th}$ and $\beta^{th}$ row is $1$.
\item $f_{pj}=1\textrm{ for exactly one }j,\:1\leq j\leq k,\qquad p=1,2,...,k,p\neq\alpha,\beta$
,i.e. except $\alpha^{th}$ and $\beta^{th}$ row every other row
have exactly one element $1$.
\end{enumerate}
\medskip{}

\begin{description}
\item [{Class3}] We define $\mathbb{C}_{3}\subseteq\mathcal{F}$ which
consists of all filter matrix $\mathbf{F}=((f_{ij}))\:,\:1\leq i,j\leq k$
such that\end{description}
\begin{enumerate}
\item $\exists$ $\alpha$ and $\beta$ such that $f_{\alpha i}=0,i=1,2,...,k$
and $f_{\beta i}=0,i=1,2,...,k$ i.e. the $\alpha^{th}$ and the $\beta^{th}$
row of $\mathbf{F}$ is zero. 
\item $f_{qi}=1\textrm{ for exactly one }i,\:1\leq i\leq k,\qquad q=1,2,...,k\,,\,q\neq\alpha,\beta$
,i.e. except $\alpha^{th}$ and $\beta^{th}$ row every other row
have exactly one element $1$.
\item $f_{j\alpha}=f_{j\beta}=1\quad1\leq j\leq k,\qquad,\,j\neq\alpha,\beta$,
i.e. except $\alpha^{th}$ and $\beta^{th}$ row every other element
of $\alpha^{th}$ and $\beta^{th}$ column is $1$.
\item $f_{jp}=1\textrm{ for exactly one }j,\:1\leq j\leq k,\qquad p=1,2,...,k,\,p\neq\alpha,\beta$
,i.e. except $\alpha^{th}$ and $\beta^{th}$ column every other column
have exactly one element $1$.
\end{enumerate}
The following theorem and its corollary provide sufficient conditions
for filter matrices to be identifiable. Any filter matrix which belong
to at least one of the three classes is identifiable. The proof of
the theorem is given in the appendix. 
\begin{thm}
\label{Theo13}Consider a simple markov chain $\mathbf{X}$ on finite
state space $\{1,2,...,k\}$ and transition probabilities $p_{ij}$
where $0<p_{ij}<1,i,j=1,2,...k.$ Suppose $\mathbf{F}$ be any filter
matrix belonging to the class $\mathbb{C}_{*}=\mathbb{C}_{1}\cup\mathbb{C}_{2}\cup\mathbb{C}_{3}$.
Then $\mathbf{F}$ must also belong to the class $\mathcal{I}.$
\end{thm}
The following corollary to the above theorem is an immediate application
of Lemma \ref{lemma9}. 
\begin{cor}
$\overline{\mathbb{C}_{*}}\subseteq\mathcal{I}$.
\end{cor}
Thus if we start with a definite structure of matrices in $\mathbb{C}_{1}$
or $\mathbb{C}_{2}$ or $\mathbb{C}_{3}$ we get a relatively larger
class $\overline{\mathbb{C}_{*}}$ of identifiable filter matrices.
For the rest of the paper we shall be working with filter matrices
within this class. We shall find that any filter matrix in this class
will provide considerable reduction in data.

\section{Structural zeroes in Transition probability matrix}

In the previous section while obtaining the sufficient conditions
for identifiability we assumed $0<p_{ij}<1,\,\forall\,i,j.$ This
was a crucial assumption in developing the theory for the sufficient
conditions. However in many practical applications this assumption
stands out to be too restrictive. For example, while modeling a disease
status the probability of an individual entering from one state to
another may be zero (in case of chronic illness, the condition of
an individual usually deteriorates). Also the case of structural zeroes
in the transition probability matrix will occur later in this paper
while dealing with multiple markov chains. In this section we generalize
the sufficient conditions for a filter matrix to be identifiable even
when some $p_{ij}$'s are zero. 

We note that all zeroes (if any) in the transition probability model
are structural zeroes, that is, we know the position of the zeroes
even before the collection of the data. Also for any $i,\:p_{ij}$
must be positive for at least one $j$ since all the row sums of the
transition probability matrix is $1.$ We further assume 

\begin{center}
(\textbf{A1}) for any $j,\:p_{ij}$ must be positive for at least
one $i.$ 
\par\end{center}

This is a reasonable assumption to make because if such a state $j$
exists we shall ignore that state from our analysis. 

As before we have the classes of filter matrix $\mathbb{C}_{1}$,
$\mathbb{C}_{2}$ and $\mathbb{C}_{3}.$ Further let us define an
additional class of filter matrix $\mathfrak{R}\subseteq\mathcal{F}$
as 

\[
\mathfrak{R}=\{\mathbf{F}\in\mathcal{F}:\textrm{ For any }i\in\{1,2,...,n\},f_{ij}=1\textrm{ for at least one }j\in Z\}
\]
where $Z=\{j:p_{ij}>0\}.$ This restriction means for every row of
a filter matrix, we should observe at least one probable transition.
The restriction on the filter matrices is quite justified and does
not in any way reduces the applicability of filtering mechanism. The
following theorem is a generalization of Theorem 13 in the case where
we allow some $p_{ij}$ to be zero. 
\begin{thm}
Consider a simple markov chain $\mathbf{X}$ on finite state space
$\{1,2,...,k\}$ and transition probabilities $p_{ij}$ where $0\leq p_{ij}\leq1,i,j=1,2,...k.$
Let $\mathbf{F}$ be any filter matrix belonging to the class $\overline{\mathcal{S}}$
where $\mathcal{S}=\mathbb{C}_{*}\cap\mathfrak{R}.$ Then under the
assumption $\textrm{A1},$ $\mathbf{F}$ must also belong to the class
$\mathcal{I}.$
\end{thm}
The proof of the above theorem is similar to the proof of Theorem
\ref{Theo13} because under the assumption A1, and for filter matrices
within the class $\mathcal{S}$, we have $P(S_{\pi})>0$ for all choices
of sequences $\pi,$ that we require in Theorem \ref{Theo13}. Finally
application of lemma \ref{lemma9}, gives the required result.

\section{Estimation and testing}

As mentioned earlier, a markov process can be completely characterized
by specifying the transition probability matrix. This section deals
with drawing inferences regarding the parameters. Instead of recording
the entire markov chain $\mathbf{x}$, we apply a given filter matrix
$\mathbf{F}\in\mathcal{F}$ to record $\phi_{\mathbf{F}}(\mathbf{x}).$
$\mathbf{F}$ is fixed and does not in any way depend on the data
$\mathbf{x}.$ The choice of $\mathbf{F}$ may depend on the availability
of the samples, storage facilities or past experience subject to the
constraint of identifiability. Based on $\phi_{\mathbf{F}}(\mathbf{x})$
we shall find estimates of the transition probabilities and compute
the standard error of the estimates. Our main tool for estimation
will be EM algorithm. For the computation of the standard error we
shall use Supplemented EM algorithm(SEM). The latter half of the section
deals with testing of hypothesis regarding the parameters.

\subsection{Estimation of parameters:}

In the present situation the complete data is $\mathbf{x}$ which
is unobserved and the observed data is $\phi_{\mathbf{F}}(\mathbf{x}).$
As a natural tool of missing data analysis we will apply EM algorithm
for the estimation of the parameter $\boldsymbol{\theta}.$ Each iteration
of EM algorithm consists of a E step (expectation step) and an M step
(maximization step). In the E step of the algorithm we need to find
the conditional expectation of the complete data log-likelihood given
the observed data and the current iterated value of the parameters.
In our case this requires to find the conditional expectation with
respect to the conditional distribution of $\mathbf{x}$ given $\phi_{\mathbf{F}}(\mathbf{x})$
and the current iterated value $\boldsymbol{\theta}^{(t)}$ of the
parameter. The complete data log likelihood is 
\[
\ell_{com}(\boldsymbol{\theta})\propto\sum_{\alpha,\beta=1}^{k}n_{\alpha\beta}\log p_{\alpha\beta}.
\]
Since $\ell_{com}(\boldsymbol{\theta})$ is linear in $n_{\alpha\beta}$,
we need to compute 
\[
E\Big(n_{\alpha\beta}|\phi_{\mathbf{F}}(\mathbf{x}),\boldsymbol{\theta}^{(t)}\Big).
\]
This conditional distribution cannot be computed directly as the conditional
distribution of $\mathbf{x}$ given $\phi_{\mathbf{F}}(\mathbf{x})$
cannot be found out explicitly. We shall express $n_{\alpha\beta}$
as a sum of certain indicator variables to evaluate this conditional
expectation, the computation of which will be shown in subsection
5.1.2. We will show that this require us to find the following conditional
probability: 
\[
P\Big(\textrm{the chain moves from state }\alpha\textrm{ to state }\beta|\phi_{\mathbf{F}}(\mathbf{x}),\boldsymbol{\theta}^{(t)}\Big).
\]
Since the observed chain has runs of missing states, the calculation
of the above probability will require us to find the probability of
a transition from one state to another in any number of steps such
that all the intermediate steps are missing. If the complete chain
is available, then the probability of a transition from $a$ to $b$
in $\nu$ steps is the $(a,b)^{th}$ element of $P^{\nu}.$ However
we need to find the probability of such transition through some specific
ways.

\subsubsection{Defining possible missing paths for a transition:}

Consider two states $a$ and $b$. Suppose we are interested in transition
from $a$ to $b$ in $\nu$ steps. Each possible way of transition
from $a$ to $b$ in $\nu$ steps is called a path of order $\nu.$
We call a path of order $1$ as edge.Thus a any given path consists
of one or more edges. Clearly the transition from $a$ to $b$ in
$\nu$ steps can occur through one or more paths. We classify these
paths in two categories based on the given filter matrix:
\begin{itemize}
\item \textbf{observed path}($\mathcal{O}$): whose all edges are observed.
\item \textbf{unobserved path}($\mathcal{U}$): whose all edges are unobserved.
\end{itemize}
Clearly the two sets $\mathcal{O}$ and $\mathcal{U}$ are not mutually
exhaustive, that is, we cannot classify all paths into any one of
these categories.
\begin{example}
Consider a two state markov chain and two filter matrices $F_{1}$
and $F_{2}$ such that 
\[
F_{1}=\begin{bmatrix}1 & 0\\
1 & 1
\end{bmatrix}\qquad\qquad F_{2}=\begin{bmatrix}1 & 0\\
0 & 1
\end{bmatrix}
\]
Suppose we consider the \emph{transition from state $1$ to state
$1$ in two steps}. The possible paths are: 
\[
w_{1}:\:1\longrightarrow1\longrightarrow1\qquad\qquad w_{2}:\:1\longrightarrow2\longrightarrow1
\]
For filter matrix $F_{1}$, path $w_{1}\in\mathcal{O}$, i.e. path
$w_{1}$is observed whereas $\mathcal{U}$ is empty, i.e. no paths
are unobserved. For filter matrix $F_{2}$, path $w_{1}\in\mathcal{O}$
and $w_{2}\in\mathcal{U}$. If we consider the \emph{transition from
state $2$ to state $2$ in two steps}, the possible paths are: 
\[
w_{1}:\:2\longrightarrow2\longrightarrow2\qquad\qquad w_{2}:\:2\longrightarrow1\longrightarrow2
\]
For filter matrix $F_{1}$, path $w_{1}\in\mathcal{O}$, and $\mathcal{U}$
is empty whereas for filter matrix $F_{2}$, path $w_{1}\in\mathcal{O}$
and $w_{2}\in\mathcal{U}$.
\end{example}
Now consider the transition probability matrix $P$ of the markov
chain. We construct two matrices $P^{[0]}=((p_{ij}^{[0]}))$ and $P^{[1]}=((p_{ij}^{[1]}))$
as  
\[
p_{ij}^{[0]}=\begin{cases}
0 & \textrm{if}\:f_{ij}=1\\
p_{ij} & \textrm{if}\:f_{ij}=0
\end{cases}
\]
and 
\[
p_{ij}^{[1]}=\begin{cases}
0 & \textrm{if}\:f_{ij}=0\\
p_{ij} & \textrm{if}\:f_{ij}=1
\end{cases}
\]
Then the $(i,j)^{th}$ element of $(P^{[0]})^{\nu}$ gives the probability
of going \emph{from state $i$ to state $j$ in }$\nu$\emph{ steps
through unobserved path(s)}. Also the $(i,j)^{th}$ element of $(P^{[1]})^{\nu}$
gives the probability of going \emph{from state $i$ to state $j$
in $\nu$ steps through observed path(s)}. 
\begin{example}
Returning to the previous example we see that for the filter matrix
$F_{1}$, 
\[
P=\begin{bmatrix}p_{11} & p_{12}\\
p_{21} & p_{22}
\end{bmatrix}\qquad P^{[0]}=\begin{bmatrix}0 & p_{12}\\
0 & 0
\end{bmatrix}\qquad P^{[1]}=\begin{bmatrix}p_{11} & 0\\
p_{21} & p_{22}
\end{bmatrix}
\]
Then 
\[
(P^{[0]})^{2}=\begin{bmatrix}0 & 0\\
0 & 0
\end{bmatrix}\qquad(P^{[1]})^{2}=\begin{bmatrix}p_{11}^{2} & 0\\
p_{21}p_{11}+p_{22}p_{21} & p_{22}^{2}
\end{bmatrix}
\]
Thus for filter matrix $F_{1}$, probability of going from any state
$i$ to any state $j$ through the unobserved paths in $2$ steps
is zero. Also 
\[
(P^{[0]})^{\nu}=\begin{bmatrix}0 & 0\\
0 & 0
\end{bmatrix}\quad\textrm{for any }\nu
\]
which means that the probability of going from any state $i$ to any
state $j$ through the unobserved paths in any steps is zero. 

Similarly for filter matrix $F_{2}$,
\[
P^{[0]}=\begin{bmatrix}0 & p_{12}\\
p_{21} & 0
\end{bmatrix}\qquad(P^{[0]})^{2}=\begin{bmatrix}p_{12}p_{21} & 0\\
0 & p_{21}p_{12}
\end{bmatrix}
\]
Thus for $F_{2}$, the probability of going from state $1$ to state
$1$ in $2$ steps through the unobserved paths is $p_{12}p_{21}$
and the probability of going from state $2$ to state $2$ in $2$
steps through the unobserved paths is $p_{21}p_{12}$. 
\end{example}
Thus given a filter matrix, the probability of going from a state
$a$ to a state $b$ in $\nu$ steps through the unobserved paths
is the $(a,b)^{th}$ element of $(P^{[0]})^{\nu}$ which is $p_{ab}^{(\nu)0}$.

\subsubsection{Estimation by EM Algorithm:}

For the $i^{th}$ transition, let, 
\[
Y_{1i}=\textrm{state from where the transition occurs}
\]
\[
Y_{2i}=\textrm{state to where the transition occurs}
\]
Thus $Y_{1i}$ and $Y_{2i}$ are two discrete random variables taking
values in the state space $\{1,2,\cdots k\}$ for all $i$. Let us
express the total number of transitions $n_{\alpha\beta}$ from the
state $\alpha$ to the state $\beta$ as $n_{\alpha\beta}={\displaystyle \sum_{i=1}^{n}I(Y_{1i}=\alpha,Y_{2i}=\beta)}$
where 
\[
I(Y_{1i}=\alpha,Y_{2i}=\beta)=\begin{cases}
1 & if\;Y_{1i}=\alpha,Y_{2i}=\beta\\
0 & \textrm{otherwise}.
\end{cases}
\]
The complete data likelihood then can be written as 
\[
L_{com}(p)\propto\prod_{i=1}^{n}f(y_{1i},y_{2i}|p)=constant\times\prod_{\alpha,\beta=1}^{k}p_{\alpha\beta}^{{\displaystyle n_{\alpha\beta}}}
\]
where 
\[
p_{\alpha k}=1-\sum_{j=1}^{k-1}p_{\alpha j}\:\forall\alpha=1(1)n.
\]
After $t$ iterations in the EM algorithm we write the E-step and
the M-step as:

\medskip{}

\textbf{\uline{E-step:}}

Let $P_{(t)}=((p_{\alpha\beta(t)}))$ be the value of the transition
probability matrix after $t$ iterations. The corresponding value
of the parameter $\boldsymbol{\theta}$ is $\boldsymbol{\theta}^{(t)}.$
The other matrices we construct take the values $P_{(t)}^{[0]}$ and
$(P_{(t)}^{[0]})^{\nu}=((p_{ab(t)}^{(\nu)0})).$ Then we compute the
expected complete data log-likelihood with respect to the conditional
distribution of $\mathbf{x}|\phi_{\mathbf{F}}(\mathbf{x}),\boldsymbol{\theta}^{(t)}.$
The complete data log-likelihood is given by 
\[
\ell_{com}(\boldsymbol{\theta})=constant+\sum_{\alpha,\beta=1}^{k}\Big\{(\log p_{\alpha\beta})\times{\displaystyle n_{\alpha\beta}}\Big\}.
\]
We then compute 
\[
Q(\boldsymbol{\theta})=E(\ell_{com}(\boldsymbol{\theta})|\phi_{\mathbf{F}}(\mathbf{x}),\boldsymbol{\theta}^{(t)}).
\]
Since $\ell_{com}(\boldsymbol{\theta})$ is linear in $n_{\alpha\beta}$,
we need to compute 
\[
E\Big(n_{\alpha\beta}|\phi_{\mathbf{F}}(\mathbf{x}),\boldsymbol{\theta}^{(t)}\Big)
\]
\[
=\sum_{i=1}^{n}E\Big(I(Y_{1i}=\alpha,Y_{2i}=\beta)|\phi_{\mathbf{F}}(\mathbf{x}),\boldsymbol{\theta}^{(t)}\Big)
\]
\[
=\sum_{i=1}^{n}P\Big(Y_{1i}=\alpha,Y_{2i}=\beta|\phi_{\mathbf{F}}(\mathbf{x}),\boldsymbol{\theta}^{(t)}\Big).
\]
Let us denote $P\Big(Y_{1i}=\alpha,Y_{2i}=\beta|\phi_{\mathbf{F}}(\mathbf{x}),\boldsymbol{\theta}^{(t)}\Big)=P_{\alpha\beta}^{i}$.
Then for each $i$, $P_{\alpha\beta}^{i}$ takes one of the three
forms $P_{\alpha\beta}^{i(1)}$, $P_{\alpha\beta}^{i(2)}$ or $P_{\alpha\beta}^{i(3)}$
as follows: 
\begin{itemize}
\item Case I ($Y_{1i}$ observed): Suppose we have a missing chain of length
$\nu-1$ with the next observed state $b.$ Then 
\[
P_{\alpha\beta}^{i}=\begin{cases}
\frac{p_{\alpha\beta}\times p_{\beta b}^{(\nu-1)0}}{p_{\alpha b}^{(\nu)0}} & if\:Y_{1i}=\alpha\\
0 & if\:Y_{1i}\neq\alpha
\end{cases}=:P_{\alpha\beta}^{i(1)}\textrm{, say.}
\]

\item Case II ($Y_{2i}$ observed): Suppose we have a missing chain of length
$\nu-1$ with the previous observed state $a.$ Then
\[
P_{\alpha\beta}^{i}=\begin{cases}
\frac{p_{a\alpha}^{(\nu-1)0}\times p_{\alpha\beta}}{p_{a\beta}^{(\nu)0}} & if\:Y_{2i}=\beta\\
0 & if\:Y_{2i}\neq\beta
\end{cases}=:P_{\alpha\beta}^{i(2)}\textrm{, say.}
\]

\item Case III (Both are not observed): Suppose we have a missing chain
of length $\nu-1$ with the previous observed state $a$ and the next
observed state $b$. Then 
\[
P_{\alpha\beta}^{i}=\frac{p_{a\alpha}^{(m)0}p_{\alpha\beta}p_{\beta b}^{(n)0}}{p_{ab}^{(\nu)0}}=:P_{\alpha\beta}^{i(3)}\textrm{, say.}
\]
where $m+n=\nu-1$ and $a=Y_{1(i-m+1)}$ and $b=Y_{2(i+n)}$. If there
is no such next observed state (that is, the observed chain ends)
then 
\[
P_{\alpha\beta}^{i}=\frac{p_{a\alpha}^{(m)0}p_{\alpha\beta}({\displaystyle \sum_{b}}p_{\beta b}^{(n)0})}{{\displaystyle \sum_{b}}p_{ab}^{(\nu)0}}=:P_{\alpha\beta}^{i(3)}\textrm{, say.}
\]

\end{itemize}
\medskip{}

\textbf{\uline{M-step:}}

We try to maximize $Q(\boldsymbol{\theta})$ with respect to $\boldsymbol{\theta}.$
Setting $\frac{\partial}{\partial\theta_{j}}Q(\boldsymbol{\theta})=0$
for each $j=1(1)d$ we get 
\[
\boldsymbol{\theta}^{(t+1)}=(\theta_{1}^{(t+1)},\theta_{2}^{(t+1)},\cdots,\theta_{d}^{(t+1)})
\]
where 
\[
\theta_{j}^{(t+1)}=\frac{{\displaystyle \sum_{l=1}^{n}P_{1j}^{l}}}{{\displaystyle \sum_{\beta=1}^{k}}{\displaystyle \sum_{l=1}^{n}P_{1\beta}^{l}}}\quad\textrm{for any }j=1,2,\cdots,(k-1)
\]
and 
\[
\theta_{(i-1)k+j}^{(t+1)}=\frac{{\displaystyle \sum_{l=1}^{n}P_{i(j+1)}^{l}}}{{\displaystyle \sum_{\beta=1}^{k}}{\displaystyle \sum_{l=1}^{n}P_{i\beta}^{l}}}\quad\textrm{for any }j=0,1,\cdots,(k-1)\textrm{ and }i=2,3,\cdots,k.
\]

\subsection{Estimation of Standard Errors:}

Since EM estimate of the parameters are the maximum likelihood estimate
of the observed likelihood, the large sample covariance matrix can
be obtained by inverting the observed information matrix. But in our
problem the observed likelihood is not known explicitly. An alternative
way is using Supplemented EM Algorithm (SEM) which allows us to find
the large sample covariance matrix without inverting the estimate
of the observed information matrix. SEM algorithm is a procedure of
obtaining a numerically stable estimate of the covariance matrix of
the estimated parameters using only the code for the steps in EM algorithm,
code for computing the large sample complete data covariance matrix
and standard matrix operations.

\subsubsection{\textbf{Supplemented EM Algorithm}:}

Since each step of the EM algorithm produces a fresh estimate of the
parameter from the previous estimates, EM algorithm can be considered
as a mapping $M$ on the parameter space. The derivative of the EM
mapping, which we call $M_{(1)}$, can be expressed in the form 
\[
M_{(1)}=i_{mis}i_{com}^{-1}=I-i_{obs}i_{com}^{-1}.
\]
The above equation implies 
\[
i_{obs}^{-1}=i_{com}^{-1}(I-M_{(1)})^{-1}
\]
which in turn implies
\[
V_{obs}=V_{com}(I-M_{(1)})^{-1}.
\]
Now we note that 
\[
V_{obs}=V_{com}(I+M_{(1)}-M_{(1)})(I-M_{(1)})^{-1}=V_{com}+\vartriangle V
\]
where $\vartriangle V=V_{com}M_{(1)}(I-M_{(1)})^{-1}$ is the increment
in variance due to missingness.

\subsubsection{\textbf{Calculation of $V_{com}$:}}

The complete data log-likelihood is 
\[
\ell_{com}(\boldsymbol{\theta})=constant+\sum_{i=1}n_{ij}\log p_{ij}\qquad\textrm{where }p_{ik}=1-\sum_{j=1}^{k-1}p_{ij}\:\forall i.
\]
\[
\therefore\quad\frac{\partial}{\partial p_{ij}}\ell_{com}=\frac{n_{ij}}{p_{ij}}-\frac{n_{ik}}{p_{ik}}
\]
Thus the gradient vector is 
\[
\mathbf{S=}\begin{bmatrix}\frac{n_{11}}{p_{11}}-\frac{n_{1k}}{p_{1k}}\\
\vdots\\
\frac{n_{k(k-1)}}{p_{k(k-1)}}-\frac{n_{kk}}{p_{kk}}
\end{bmatrix}
\]
Now for $i\neq i'\qquad\frac{\partial^{2}}{\partial p_{ij}\partial p_{i'j'}}\ell_{com}=0$.
Also we have 
\[
\frac{\partial^{2}}{\partial p_{ij}\partial p_{ij'}}\ell_{com}=\begin{cases}
\frac{n_{ik}}{2p_{ik}^{2}} & \quad\textrm{if }j\neq j'\\
\frac{1}{2}\Bigg[\frac{n_{ik}}{p_{ik}^{2}}-\frac{n_{ij}}{p_{ij}^{2}}\Bigg] & \quad\textrm{if }j=j'
\end{cases}
\]
Let $B$ be the matrix of the negatives of the second order derivatives.
Then $B$ is a matrix of order $k^{2}-k$ such that $B=blockdiagonal(B_{1},B_{2},\cdots B_{k})$
where $B_{i}=((b_{jj'}^{i}))_{k-1}$ where $b_{jj'}^{i}=-\frac{\partial^{2}}{\partial p_{ij}\partial p_{ij'}}\ell_{com}.$ 

Then the fisher information matrix of the complete data is 
\[
i_{com}=E(B\,|\boldsymbol{\theta},data)=blockdiagonal\Big(E(B_{1}),E(B_{2}),\cdots,E(B_{k})\Big)
\]
where $E(B_{i})=((E(b_{jj'}^{i}|\boldsymbol{\theta},data)))$. Thus
the variance-covariance matrix of the complete data is $V_{com}=i_{com}^{-1}.$

\subsubsection{\textbf{Computing $M_{(1)}$ by numerical differentiation:}}

For our problem the mapping $M=M(\theta_{1},\theta_{2},\cdots,\theta_{d}):\boldsymbol{\Theta}\rightarrow\boldsymbol{\Theta}$
is not known explicitly. The derivative of $M$ at $\hat{\boldsymbol{\theta}}$
is calculated numerically from the output of the forced EM steps.
$M_{(1)}$ is the matrix with the $(i,j)^{th}$element as $\frac{\vartriangle M_{j}}{\vartriangle\hat{\theta}_{i}}=$change
in the $j^{th}$ component of $M$ due to the change in the $i^{th}$element
of $\hat{\boldsymbol{\theta}}$. For this we start with the EM estimate
$\hat{\boldsymbol{\theta}}$ and change its $i^{th}$ element $\hat{\theta}_{i}$
by $\theta_{i}^{(t)}$. We call this resultant estimate by $\boldsymbol{\theta}^{(t)}(i)$
and run one EM iteration on it to get $\tilde{\boldsymbol{\theta}}^{(t+1)}(i)$.
Then $\vartriangle M_{j}=\tilde{\theta}_{j}^{(t+1)}(i)-\hat{\theta}_{j}$
and $\vartriangle\hat{\theta}_{i}=\theta_{i}^{(t)}-\hat{\theta}_{i}$
and so we compute the ratio $r_{ij}=\frac{\vartriangle M_{j}}{\vartriangle\hat{\theta}_{i}}.$
Thus we run a sequence of SEM iterations, where the $(t+1)^{th}$
iteration is defined as follows:

\begin{algorithm}
\caption{SEM Algorithm}

We take as input $\hat{\boldsymbol{\theta}}$and $\boldsymbol{\theta}^{(t)}$. 
\begin{enumerate}
\item Run the usual E step and M steps to get $\boldsymbol{\theta}^{(t+1)}$;
\item Fix $i=1$. Calculate 
\[
\boldsymbol{\theta}^{(t)}(i)=(\hat{\theta}_{1},\cdots\hat{\theta}_{i-1},\theta_{i}^{(t)},\hat{\theta}_{i+1},\cdots,\hat{\theta}_{d})
\]
which is $\hat{\boldsymbol{\theta}}$ except the $i^{th}$ component
which equals $\theta_{i}^{(t)}.$
\item Treating $\boldsymbol{\theta}^{(t)}(i)$ as the current estimate of
$\boldsymbol{\theta}$, run one iteration of EM to obtain $\tilde{\boldsymbol{\theta}}^{(t+1)}(i)$.
\item Obtain the ratio 
\[
r_{ij}^{(t)}=\frac{\tilde{\theta}_{j}^{(t+1)}(i)-\hat{\theta}_{j}}{\theta_{i}^{(t)}-\hat{\theta}_{i}}\qquad\textrm{for }j=1,2,\cdots,d.
\]

\item Repeat steps 2 to 4 for $i=1,2,\cdots,d$.
\end{enumerate}
We get as output $\boldsymbol{\theta}^{(t+1)}$ and $\{r_{ij}^{(t)}:\;i,j=1,2,\cdots,d\}$.

$M_{(1)}$ is the limiting matrix $\{r_{ij}\}$ as $t\rightarrow\infty$.
\end{algorithm}

A difficulty in running the SEM iterations is that while changing
the $i^{th}$ element $\hat{\theta}_{i}$ by $\theta_{i}^{(t)}$ the
resultant estimate $\boldsymbol{\theta}^{(t)}(i)$ may not belong
to the parameter space $\boldsymbol{\Theta}$ because the sum of the
corresponding row probabilities ${\displaystyle \sum}_{j=1}^{k-1}p_{ij}$
may be more than $1$. Thus theoretically the mapping $M$ may not
be defined in such cases. Then we replace $\theta_{i}^{(t)}$ by $\theta_{i}^{(t)}-\epsilon\;,(\epsilon>0)$
so that the corresponding sum of probability is less than $1.$

\subsubsection{\textbf{Implementational Issues:}}

While implementing the SEM algorithm it is always safe to start with
the initial values of the original EM algorithm for numerical accuracy.
But this may result in too many unnecessary iterations because the
initial choice may be too far from the MLE. Hence Meng and Rubin suggested
to take the initial choice in SEM as a suitable iterate of the EM
algorithm or two complete data standard deviations from the MLE. Computation
of $M_{(1)}$ being numerical differentiation is less accurate than
evaluating the function $M$ itself. Hence the stopping criterion
should be less stringent for SEM algorithm as compared to the original
EM algorithm. Meng and Rubin suggested to use square root of the stopping
criterion of the original EM as the stopping criterion for SEM.

The observed variance covariance matrix obtained by SEM algorithm
should be theoretically a real symmetric positive definite matrix.
This provide a diagnostics for programming errors and numerical precision.
The numerical symmetry of the final matrix increases with more stringent
criterion in the algorithm.

\subsection{Testing of Hypothesis}

The large sample inferences on the EM estimate can be drawn using
the asymptotic distribution 
\[
\hat{\boldsymbol{\theta}}\sim N(\boldsymbol{\theta},V_{obs})
\]

Since SEM algorithm helps us to numerically estimate $V_{obs}$, we
can use the above distribution for testing of the parameters and finding
confidence intervals.

\subsubsection{Testing the transition probability matrix}

Suppose we wish to test the hypothesis $H_{0}:P=P_{0}$. Since only
$k(k-1)$ parameters of the transition probability matrix are independent,
the above hypothesis is equivalent to $H_{0}:\boldsymbol{\theta}=\boldsymbol{\theta}_{0}$.
Now 
\[
(\hat{\boldsymbol{\theta}}-\boldsymbol{\theta})'V_{obs}^{-1}(\hat{\boldsymbol{\theta}}-\boldsymbol{\theta})\sim\chi_{k^{2}}^{2}
\]
which implies the test statistic for testing $H_{0}$ is $\chi^{2}=(\hat{\boldsymbol{\theta}}-\boldsymbol{\theta}_{0})'V_{obs}^{-1}(\hat{\boldsymbol{\theta}}-\boldsymbol{\theta}_{0})$
which has $\chi_{k^{2}}^{2}$ distribution under $H_{0}.$ Thus the
critical region for testing $H_{0}$ is $\{\mathbf{x}:\chi^{2}>\chi_{k^{2},\alpha}^{2}\}$

\subsubsection{Test of hypotheses about specific probabilities and confidence regions}

First we consider testing the hypothesis that certain transition probabilities
$p_{ij}$ have specified values $p_{ij}^{0}.$ Under the null hypothesis
$H_{0i}:\theta_{i}=\theta_{i}^{0}$, the statistic $\tau_{i}=\frac{\hat{\theta}_{i}-\theta_{i}^{0}}{\sqrt{s_{ii}}}$
has $N(0,1)$ distribution. Thus the critical region for testing $H_{0i}$
is $\{\mathbf{x}:|\tau_{i}|>z_{\alpha/2}\}$. The $100(1-\alpha)\%$
confidence interval for $\theta_{i}$ is $(\hat{\theta}_{i}-\sqrt{s_{ii}}z_{\alpha/2},\hat{\theta}_{i}+\sqrt{s_{ii}}z_{\alpha/2})$.

\section{Multiple Markov chains}

Let $\{X_{n}\}$ be a $s^{th}$ order markov chain. In the previous
sections we have discussed the case where $s=1,$ that is, simple
markov chains. If $s>1$, then $\{X_{n}\}$ is called a multiple markov
chain of order $s$ with transition probabilities
\[
p_{a_{1},\cdots,a_{s}:a_{s+1}}=P(X_{n}=a_{s+1}|\,X_{n-1}=a_{s},X_{n-2}=a_{s-1},\cdots,X_{n-s}=a_{1}).
\]
Multiple markov chains of any order can be reduced to a simple markov
chain by the following technique. 

Suppose $\{X_{n}\}$ is called a markov chain of order $s$ with $k$
states. We define a new stochastic process $\{Y_{n},n=1,2,\cdots\}$
where $Y_{n}=(X_{n},X_{n+1},\cdots,X_{n+s-1}).$ Then $\{Y_{n}\}$
is a simple markov chain whose state space has $k^{s}$ different
$s-$tuples. The transition probabilities of the new defined markov
process are
\[
p_{(a_{1},a_{2},\cdots,a_{s})(b_{1},b_{2},\cdots,b_{s})}=\begin{cases}
p_{a_{1}a_{2}\cdots a_{s}:b_{s}} & \textrm{if }\:b_{i}=a_{i+1},i=1,2,\cdots,s-1\\
0 & \textrm{otherwise.}
\end{cases}
\]
The number of positive entries in the $k^{s}\times k^{s}$ transition
probability matrix is $k^{s+1}.$ The parameters of interest are the
probabilities $p_{a_{1},\cdots,a_{s}:a_{s+1}}$ which requires estimation
from the data. 

In this situation we apply our filtering technique to the chain $\{y_{n}\}.$
But now the transition probability matrix contains many zero elements
and hence the additional restriction described section $4$ needs
to be applied on the filter matrices. We note that in this case the
transition probability matrix satisfies the assumption made in section
$5.$ The technique of estimation of the parameters from the data
$\phi_{\mathbf{F}}(\mathbf{y})$ remains same as in the simple markov
chain.

\section{Simulation Study}

For simulation we start with a markov chain with $3$ states. A markov
chain of length $1000$ is being generated with the transition probability
matrix 

\[
P=\begin{bmatrix}0.2 & 0.3 & 0.5\\
0.8 & 0.1 & 0.1\\
0.7 & 0.1 & 0.2
\end{bmatrix}.
\]
The filter matrix for generating the observed chain is 
\[
F=\begin{bmatrix}0 & 1 & 0\\
1 & 1 & 0\\
1 & 0 & 0
\end{bmatrix}.
\]
Clearly the filter matrix used satisfy the sufficient condition for
estimability. With this filter matrix we reduce $16\%$ of the data,
i.e., from the complete markov chain of length $1000$ we do not observe
$16\%$ of the data. The precision we use in estimating the parameters
through the steps of the EM algorithm is of the order $10^{-12}$
and the precision used in computing the standard error is of the order
$10^{-6}.$ With this precision the estimated transition probability
matrix is 
\[
\hat{P}=\begin{bmatrix}0.2411168 & 0.2850831 & 0.4738001\\
0.7395851 & 0.1429865 & 0.1174284\\
0.7648870 & 0.1067367 & 0.1283763
\end{bmatrix}.
\]
The observed variance covariance matrix $V_{obs}$ as computed by
the SEM algorithm is 
\[
\begin{bmatrix}8.00\times10^{-4} & -3.00\times10^{-4} & 2.36\times10^{-5} & 4.55\times10^{-6} & 7.13\times10^{-4} & 7.78\times10^{-5}\\
-3.00\times10^{-4} & 4.70\times10^{-4} & -8.85\times10^{-6} & -1.71\times10^{-6} & -2.68\times10^{-4} & -2.92\times10^{-5}\\
2.36\times10^{-5} & -8.84\times10^{-6} & 1.01\times10^{-3} & -5.10\times10^{-4} & -8.05\times10^{-7} & -5.30\times10^{-5}\\
4.60\times10^{-6} & -1.73\times10^{-6} & -5.10\times10^{-4} & 6.06\times10^{-4} & -5.92\times10^{-8} & -1.02\times10^{-5}\\
7.13\times10^{-4} & -2.68\times10^{-4} & -7.52\times10^{-7} & -1.45\times10^{-7} & 1.80\times10^{-3} & -1.06\times10^{-4}\\
7.78\times10^{-5} & -2.92\times10^{-5} & -5.30\times10^{-5} & -1.02\times10^{-5} & -1.06\times10^{-4} & 3.92\times10^{-4}
\end{bmatrix}.
\]
The complete data variance covariance matrix $V_{com}$ is 
\[
\begin{bmatrix}0.0003674 & -0.0001380\\
-0.0001380 & 0.0004092\\
 &  & 0.0009496 & -0.0005214\\
 &  & -0.0005214 & 0.0006041\\
 &  &  &  & 0.0006033 & -0.0002739\\
 &  &  &  & -0.0002739 & 0.0003199
\end{bmatrix}
\]
The increase in variance $\triangle V$ is 
\[
\begin{bmatrix}4.33\times10^{-4} & -1.63\times10^{-4} & 2.36\times10^{-5} & 4.58\times10^{-6} & 7.13\times10^{-4} & 7.78\times10^{-5}\\
-1.63\times10^{-4} & 6.11\times10^{-5} & -8.86\times10^{-6} & -1.71\times10^{-6} & -2.68\times10^{-4} & -2.92\times10^{-5}\\
2.35\times10^{-5} & -8.84\times10^{-6} & 5.75\times10^{-5} & 1.11\times10^{-5} & -8.06\times10^{-7} & -5.30\times10^{-5}\\
4.61\times10^{-6} & -1.73\times10^{-6} & 1.11\times10^{-5} & 2.15\times10^{-6} & -5.92\times10^{-8} & -1.02\times10^{-5}\\
7.13\times10^{-4} & -2.68\times10^{-4} & -7.52\times10^{-7} & -1.45\times10^{-7} & 1.20\times10^{-3} & 1.68\times10^{-4}\\
7.78\times10^{-5} & -2.92\times10^{-5} & -5.30\times10^{-5} & -1.02\times10^{-5} & 1.68\times10^{-4} & 7.22\times10^{-5}
\end{bmatrix}
\]

\section{Practical example}

The data consists of the daily rainfall, measured in millimeters times
10, at Alofi in the Niue Island group. 1096 observations were recorded
from 1st January 1987 until 31st December 1989. The data is classified
into three states: state 1 which represents ``no rain'', state 2
which represents ``from non zero until 5mm'' and state 3 which represents
``more than 5mm'' rain. This time series data can be considered
as a 3 state markov chain. P. J. Avery and D. A. Henderson (1999)
used this dataset for the fitting of markov model. 

For the generation of the observed data we use the same filter matrix
as in case of the simulated data. From 1096 observations we find that
this filter matrix leads to a missingness of 45.35\%. While storing
only 54.65\% of the original data we find the estimate of the transition
probability matrix is 

\[
\begin{bmatrix}0.6717154 & 0.2231926 & 0.1050920\\
0.4585938 & 0.3034812 & 0.2379251\\
0.2137608 & 0.3447883 & 0.4414509
\end{bmatrix}.
\]
We compute the observed variance covariance matrix as 
\[
\begin{bmatrix}4.65\times10^{-4} & -3.16\times10^{-4} & 1.23\times10^{-6} & 8.14\times10^{-7} & 1.15\times10^{-4} & 1.81\times10^{-4}\\
-3.16\times10^{-4} & 3.41\times10^{-4} & -8.51\times10^{-7} & -5.63\times10^{-7} & -7.76\times10^{-5} & -1.23\times10^{-4}\\
1.23\times10^{-6} & -8.51\times10^{-7} & 9.21\times10^{-4} & -4.18\times10^{-4} & -4.76\times10^{-5} & -3.09\times10^{-4}\\
8.14\times10^{-7} & -5.63\times10^{-7} & -4.18\times10^{-4} & 7.49\times10^{-4} & -3.15\times10^{-5} & -2.05\times10^{-4}\\
1.15\times10^{-4} & -7.76\times10^{-5} & -4.76\times10^{-5} & -3.15\times10^{-5} & 9.26\times10^{-4} & 1.51\times10^{-4}\\
1.81\times10^{-4} & -1.23\times10^{-4} & -3.09\times10^{-4} & -2.05\times10^{-4} & 1.51\times10^{-4} & 0.0026
\end{bmatrix}.
\]
The complete data variance covariance matrix $V_{com}$ is 
\[
\begin{bmatrix}0.000391 & -0.000266\\
-0.000266 & 0.000307\\
 &  & 0.000837 & -0.000469\\
 &  & -0.000469 & 0.0007128\\
 &  &  &  & 0.0007185 & -0.000315\\
 &  &  &  & -0.000315 & 0.0009658
\end{bmatrix}.
\]
The increase in variance due to missingness is 
\[
\begin{bmatrix}7.48\times10^{-5} & -5.00\times10^{-5} & 1.23\times10^{-6} & 8.14\times10^{-7} & 1.15\times10^{-4} & 1.81\times10^{-4}\\
-5.00\times10^{-5} & 3.38\times10^{-5} & -8.50\times10^{-7} & -5.63\times10^{-7} & -7.76\times10^{-5} & -1.23\times10^{-4}\\
1.23\times10^{-6} & -8.50\times10^{-7} & 8.39\times10^{-5} & 5.55\times10^{-5} & -4.76\times10^{-5} & -3.09\times10^{-4}\\
8.14\times10^{-7} & -5.63\times10^{-7} & 5.55\times10^{-5} & 3.68\times10^{-5} & -3.15\times10^{-5} & -2.05\times10^{-4}\\
1.15\times10^{-4} & -7.76\times10^{-5} & -4.76\times10^{-5} & -3.15\times10^{-5} & 2.08\times10^{-5} & 4.66\times10^{-4}\\
1.81\times10^{-4} & -1.23\times10^{-4} & -3.09\times10^{-4} & -2.05\times10^{-4} & 4.66\times10^{-4} & 0.0016079
\end{bmatrix}.
\]

\section{Appendix}

\subsection*{Proof of Theorem 13:}
\begin{proof}
We split the proof in three parts. We shall prove $\mathbb{C}_{i}\subseteq\mathcal{I},i=1,2,3.$
This will imply that $\mathbb{C}_{*}\subseteq\mathcal{I}.$ 

\textbf{\uline{Part 1}}:

Suppose a filter matrix $\mathbf{M}\in\mathbb{C}_{1}$. Then the $\alpha^{th}$
row and $\beta^{th}$column of $\mathbf{M}$ are zero and all other
rows and columns of $\mathbf{M}$ have exactly one element nonzero.

\textbf{Case a:} $\alpha\neq\beta$ 

\medskip{}

\textbf{Step 1}:

Consider $p_{ij}\quad1\leq i,j\leq k\:\;,i,j\neq\alpha,\beta$. 

Let the $i^{th}$column has a element $f_{ai}=1$ and let the $j^{th}$
row has a element $f_{jb}=1$. 

Then $P(S_{aijb})>0$. This implies $P(S_{ij})>0.$

Hence corollary 12 implies that $p_{ij}\quad1\leq i,j\leq k\:\;,i,j\neq\alpha,\beta$
are estimable.

\medskip{}

\textbf{Step 2:}

Next from the $\beta^{th}$ column of the transition probability matrix
consider $p_{i\beta},\forall i=1(1)k,\:i\neq\beta$.

Since $\beta\neq\alpha$, we have a $j$ such that $f_{\beta j}=1$ 

Also since $i\neq\beta$ we have $a$ such that $f_{ai}=1$

Then $P(S_{ai\beta j})>0$. This implies $P(S_{i\beta})>0.$

Again corollary 12 implies that $p_{i\beta},\forall i=1(1)k,\:i\neq\beta$
are estimable.

\medskip{}

\textbf{Step 3:}

Next from the $\alpha^{th}$ row of the transition probability matrix
consider $p_{\alpha j},\forall j=1(1)k,\:j\neq\alpha$.

For $p_{\alpha j}$ choose $i$ and $r$ such that $f_{i\alpha}=1\quad i\neq\alpha$
and $f_{jr}=1$ 

Then $P(S_{i\alpha jr})>0$. This implies $P(S_{\alpha j})>0.$

From corollary 12 we get $p_{\alpha j},\forall j=1(1)k,\:j\neq\alpha$
is estimable.

\textbf{Step 4:}

The parameter $p_{\alpha\alpha}$ is estimable from the condition
${\displaystyle \sum_{j}}p_{\alpha j}=1$

\textbf{Step 5:}

From the $\beta^{th}$ row of the transition probability matrix consider
$p_{\beta j},\forall j=1(1)k,\:j\neq\alpha$.

If $j$ is such that $f_{\beta j}=1$ then we get that $p_{\beta j}$
is estimable. Hence we now consider $j$ to be such that $f_{\beta j}=0$.

For this we now choose any state $a$ and a state $s$ such that $f_{js}=1$. 

Then $P(S_{a\:\_\:js})>0$ which implies $P(S_{\pi})>0$ where $\pi=a\:\_\:js$.

Let $C=\{b:f_{ab}=0\quad,f_{bj}=0\}$. We note that $\beta\in C$
and $p_{\pi}$ is of the form 

\[
p_{\pi}=(\sum_{b\in C,b\neq\beta}p_{ab}p_{bj}+p_{a\beta}p_{\beta j})\times p_{js}
\]
Now since $P(S_{\pi})>0,$ lemma 11 implies that $p_{\pi}$ is identifiable.
Hence 
\[
p_{\pi}=(\sum_{b\in C,b\neq\beta}p_{ab}p_{bj}+p_{a\beta}p_{\beta j})\times p_{js}=\textrm{Known Constant}
\]

Since all $p_{ab}$ and $p_{bj}$ and also $p_{a\beta}$ are identifiable
, we get $p_{\beta j},\forall j=1(1)k,\:j\neq\alpha$ are estimable.

\textbf{Step 6:}

From the $\alpha^{th}$ column of the transition probability matrix
consider $p_{i\alpha},\forall i=1(1)k,\:i\neq\beta$

If $i$ is such that $f_{i\alpha}=1$ then we get that $p_{i\alpha}$
is estimable. Hence we now consider $i$ to be such that $f_{i\alpha}=0$.

For this we now choose any state $b$ and a state $r$ such that $f_{ri}=1$

Then $P(S_{ri\:\_\:b})>0$ which implies $P(S_{\pi})>0$ where $\pi=ri\:\_\:b$.

Let $D=\{a:f_{ab}=0\quad,f_{ia}=0\}$. We note that $\alpha\in D$
and $p_{\pi}$ is of the form 
\[
p_{\pi}=p_{ri}\times(\sum_{a\in D,a\neq\alpha}p_{ia}p_{ab}+p_{i\alpha}p_{\alpha b})
\]
Now since $P(S_{\pi})>0,$ lemma 11 implies that $p_{\pi}$ is identifiable.
Hence 
\[
p_{\pi}=p_{ri}\times(\sum_{a\in D,a\neq\alpha}p_{ia}p_{ab}+p_{i\alpha}p_{\alpha b})=\textrm{Known Constant}
\]

Since all $p_{ab}$ and $p_{ia}$ and also $p_{\alpha b}$ are identifiable
, we get $p_{i\alpha},\forall i=1(1)k,\:i\neq\beta$ are estimable.

\textbf{Step 7:}

$p_{\beta\alpha}$ is estimable from the condition ${\displaystyle \sum_{j}}p_{\beta j}=1$

\textbf{\uline{Case b}}\uline{ }: $\alpha=\beta$ 

\textbf{Step 1}:

Consider $p_{ij}\quad1\leq i,j\leq k\:\;,i,j\neq\alpha$. 

The estimatibility of $p_{ij}$ is same as step 1 of \textbf{case
(a).}

\textbf{Step 2:}

Next from the $\alpha^{th}$ column of the transition probability
matrix consider $p_{i\alpha},\forall i=1(1)k,\:i\neq\alpha$.

The parameter $p_{i\alpha}$ is identified from the condition ${\displaystyle \sum_{j}}p_{ij}=1$

\textbf{Step 3:}

Next from the $\alpha^{th}$ row of the transition probability matrix
consider $p_{\alpha j},\forall j=1(1)k,\:j\neq\alpha$

For $p_{\alpha j}$ choose $i$ and $r$ such that $f_{i\alpha}=0\quad i\neq\alpha$
and $f_{jr}=1$ .

Then $P(S_{i\:\_\:jr})>0$ which implies $P(S_{\pi})>0$ where $\pi=i\:\_\:jr$.

Let $D=\{b:f_{bj}=0\quad\textrm{and}\quad f_{ib}=0\}$. We note that
$\alpha\in D$ and $p_{\pi}$ is of the form

\[
p_{\pi}=(\sum_{b\in D,b\neq\alpha}p_{ib}p_{bj}+p_{i\alpha}p_{\alpha j})p_{jr}
\]
Now since $P(S_{\pi})>0$ lemma 11 implies that $p_{\pi}$ is identifiable.
Hence
\[
p_{\pi}=(\sum_{b\in D,b\neq\alpha}p_{ib}p_{bj}+p_{i\alpha}p_{\alpha j})p_{jr}=\textrm{Known Constant}
\]

Since each of $p_{ib},p_{bj},p_{i\alpha}$ in the above equation are
already identifiable, we get that $p_{\alpha j}$ can also be identified
uniquely.

\medskip{}

\textbf{Step 4 :}

The parameter $p_{\alpha\alpha}$ is identified from the condition
${\displaystyle \sum_{j}}p_{\alpha j}=1$

Thus all the parameters for $\mathbf{M}$ are identifiable. Hence
for any matrix $\mathbf{M}\in\mathbb{C}_{1}$, we have $\mathbf{M}\in\mathcal{F}$.
Thus $\mathbb{C}_{1}\subseteq\mathcal{I}$.

\textbf{Part 2}:

In the next case, suppose a filter matrix $\mathbf{M}\in\mathbb{C}_{2}$.
Then the $\alpha^{th}$ and $\beta^{th}$column of a filter matrix
$\mathbf{M}$ are zero and all other columns of $\mathbf{M}$ have
exactly one element nonzero.

\textbf{Step 1:}

Consider $p_{ij}\quad1\leq i,j\leq k\:\;,i\neq\alpha,\beta$. 

Let the $i^{th}$ column has a element $f_{ai}=1$ and let the $j^{th}$
row has a element $f_{jr}=1.$

Then $P(S_{aijr})>0.$ This implies $P(S_{ij})>0.$

Hence applying the corollary 12 $p_{ij}\quad1\leq i,j\leq k\:\;,i\neq\alpha,\beta$
are estimable.

\textbf{Step 2:}

Consider $p_{\alpha\alpha}$ and $p_{\beta\alpha}.$ 

Let the $\alpha^{th}$ row has a element $f_{\alpha r}=1\:,r\neq\alpha,\beta$
and we choose a $i$ such that $i\neq\alpha,\beta.$

Then $P(S_{i\:\_\:\alpha r})>0$ which means $P(S_{\pi})>0$ where
$\pi=i\:\_\;\alpha r.$

Let $D=\{b:f_{ib}=0\quad\textrm{and}\quad f_{b\alpha}=0\}$. Then
$p_{\pi}$ is of the form 

\[
p_{\pi}=(\sum_{b\in D}p_{ib}p_{b\alpha})p_{\alpha r}
\]

Clearly $\alpha,\beta\in D$ and hence we get 
\[
p_{\pi}=(\sum_{b\in D,b\neq\alpha,\beta}p_{ib}p_{b\alpha}+p_{i\alpha}p_{\alpha\alpha}+p_{i\beta}p_{\beta\alpha})p_{\alpha r}
\]
Since $P(S_{\pi})>0,$ lemma 8 implies that $p_{\pi}$ is identifiable.
Hence 
\[
p_{\pi}=(\sum_{b\in D,b\neq\alpha,\beta}p_{ib}p_{b\alpha}+p_{i\alpha}p_{\alpha\alpha}+p_{i\beta}p_{\beta\alpha})p_{\alpha r}=\textrm{Known Constant}
\]

Since $p_{ib},b\neq\alpha$ and $p_{b\alpha},b\neq\alpha,\beta$ and
$p_{\alpha r},r\neq\alpha,\beta$ are all estimable from the above
equation we get a equation of the form

This gives us a equation of the form 

\[
C_{1}p_{\alpha\alpha}+C_{2}p_{\beta\alpha}=K_{1}
\]
where $C_{i}'s$ are constants.

Also we from the condition ${\displaystyle \sum_{j}}p_{ij}=1$ , since
all other parameters are estimable we get a equation of the form 
\[
p_{\alpha\alpha}+p_{\beta\alpha}=K_{2}
\]

These two equations make $p_{\alpha\alpha}$and $p_{\beta\alpha}$
estimable.

\textbf{Step 3:}

Consider $p_{\alpha\beta}$ and $p_{\beta\beta}.$

Let the $\beta^{th}$ row has a element $f_{\beta r}=1\:,r\neq\alpha,\beta$
and we choose a $i$ such that $i\neq\alpha,\beta.$

Then $P(S_{i\:\_\:\beta r})>0$ which means $P(S_{\pi})>0$ where
$\pi=i\:\_\;\beta r.$

Let $D=\{b:f_{ib}=0\quad\textrm{and}\quad f_{b\beta}=0\}$. Then $p_{\pi}$
is of the form 

\[
p_{\pi}=(\sum_{b\in D}p_{ib}p_{b\beta})p_{\beta r}
\]

Clearly $\alpha,\beta\in D$ and hence we get 
\[
p_{\pi}=(\sum_{b\in D,b\neq\alpha,\beta}p_{ib}p_{b\beta}+p_{i\alpha}p_{\alpha\beta}+p_{i\beta}p_{\beta\beta})p_{\beta r}
\]
Since $P(S_{\pi})>0,$ lemma 11 implies that $p_{\pi}$ is identifiable.
Hence
\[
p_{\pi}=(\sum_{b\in D,b\neq\alpha,\beta}p_{ib}p_{b\beta}+p_{i\alpha}p_{\alpha\beta}+p_{i\beta}p_{\beta\beta})p_{\beta r}=\textrm{Known Constant}
\]

Since $p_{ib},b\neq\alpha$ and $p_{b\beta},b\neq\alpha,\beta$ and
$p_{\beta r},r\neq\alpha,\beta$ are all estimable from the above
equation we get a equation of the form

\[
C_{1}p_{\alpha\beta}+C_{2}p_{\beta\beta}=K_{1}
\]
where $C_{1}$ and $C_{2}$ and $K_{1}$ are constants. 

Also we from the condition ${\displaystyle \sum_{j}}p_{ij}=1$ , since
all other parameters are estimable we get a equation of the form 
\[
p_{\alpha\beta}+p_{\beta\beta}=K_{2}
\]

These two equations make $p_{\alpha\beta}$and $p_{\beta\beta}$ estimable.

Thus all the parameters for $\mathbf{M}$ are identifiable. Hence
for any matrix $\mathbf{M}\in\mathbb{C}_{2}$, we have $\mathbf{M}\in\mathcal{F}$.
Thus $\mathbb{C}_{2}\subseteq\mathcal{I}$.

\textbf{Part 3}:

Now suppose a filter matrix $\mathbf{M}\in\mathbb{C}_{3}$. Then the
$\alpha^{th}$ and $\beta^{th}$row of a filter matrix $\mathbf{M}$
are zero and all other rows of $\mathbf{M}$ have exactly one element
nonzero.

\textbf{Step 1:}

Consider $p_{ij}\quad1\leq i,j\leq k\:\;,j\neq\alpha,\beta$. 

Let the $i^{th}$ column has a element $f_{ai}=1$ and let the $j^{th}$
row has a element $f_{jr}=1.$

Then $P(S_{aijr})>0.$ This implies $P(S_{ij})>0.$

Hence applying the corollary 12 $p_{ij}\quad1\leq i,j\leq k\:\;,j\neq\alpha,\beta$
are estimable.

\textbf{Step 2:}

Consider $p_{\alpha j}\:\;j=\alpha,\beta$. 

Let the $\alpha^{th}$ column has a element $f_{i\alpha}=1\:,i\neq\alpha,\beta$
and we choose a $r$ such that $r\neq\alpha,\beta.$

Then $P(S_{i\alpha\:\_\:r})>0$ which means $P(S_{\pi})>0$ where
$\pi=i\alpha\:\_\;r.$

Let $D=\{b:f_{\alpha b}=0\quad\textrm{and}\quad f_{br}=0\}.$ Then
$p_{\pi}$ is of the form 

\[
p_{\pi}=p_{i\alpha}(\sum_{b\in D}p_{\alpha b}p_{br})
\]

Clearly $\alpha,\beta\in D$ and hence we get 
\[
p_{\pi}=(\sum_{b\in D,b\neq\alpha,\beta}p_{\alpha b}p_{br}+p_{\alpha\alpha}p_{\alpha r}+p_{\alpha\beta}p_{\beta r})p_{i\alpha}
\]

Since $P(S_{\pi})>0,$ lemma 11 implies that $p_{\pi}$ is identifiable.
Hence
\[
p_{\pi}=(\sum_{b\in D,b\neq\alpha,\beta}p_{\alpha b}p_{br}+p_{\alpha\alpha}p_{\alpha r}+p_{\alpha\beta}p_{\beta r})p_{i\alpha}=\textrm{Known Constant}
\]

Since $p_{i\alpha},i\neq\alpha$ and $p_{\alpha b},b\neq\alpha,\beta$
and $p_{\beta r},r\neq\alpha,\beta$ and $p_{ab},\quad a,b\neq\alpha,\beta$
are all estimable from the above equation we get a equation of the
form 
\[
C_{1}p_{\alpha\alpha}+C_{2}p_{\alpha\beta}=K_{1}
\]
 where $C_{i}$ and $K_{i}$ are constants.

Also from the restriction ${\displaystyle \sum_{j}}p_{\alpha j}=1$
we get a equation of the form 
\[
p_{\alpha\alpha}+p_{\alpha\beta}=K_{2}
\]

These two final equations make the parameters $p_{\alpha\alpha}$
and $p_{\alpha\beta}$ identifiable.

\textbf{Step 3:}

Consider $p_{\beta j}\:\;j=\alpha,\beta$. 

Let the $\beta^{th}$ column has a element $f_{i\beta}=1\:,i\neq\alpha,\beta$
and we choose a $r$ such that $r\neq\alpha,\beta.$

Then $P(S_{i\beta\:\_\:r})>0$ which means $P(S_{\pi})>0$ where $\pi=i\beta\:\_\;r.$

Let $D=\{b:f_{\beta b}=0\quad\textrm{and}\quad f_{br}=0\}.$ Then
$p_{\pi}$ is of the form 

\[
p_{\pi}=p_{i\beta}(\sum_{b\in D}p_{\beta b}p_{br})
\]

Clearly $\alpha,\beta\in D$ and hence we get
\[
p_{\pi}=(\sum_{b\in D,b\neq\alpha,\beta}p_{\beta b}p_{br}+p_{\beta\alpha}p_{\alpha r}+p_{\beta\beta}p_{\beta r})p_{i\beta}
\]

Since $P(S_{\pi})>0,$ lemma 11 implies that $p_{\pi}$ is identifiable.
Hence
\[
p_{\pi}=(\sum_{b\in D,b\neq\alpha,\beta}p_{\beta b}p_{br}+p_{\beta\alpha}p_{\alpha r}+p_{\beta\beta}p_{\beta r})p_{i\beta}=\textrm{Known Constant}
\]

Since $p_{i\beta},i\neq\alpha$ and $p_{\beta b},b\neq\alpha,\beta$
and $p_{\alpha r},r\neq\alpha,\beta$ and $p_{ab},\quad a,b\neq\alpha,\beta$
are all estimable from the above equation we get a equation of the
form 
\[
C_{1}p_{\beta\alpha}+C_{2}p_{\beta\beta}=K_{1}
\]
 where $C_{i}$ and $K_{i}$ are constants.

Also from the restriction ${\displaystyle \sum_{j}}p_{\beta j}=1$
we get a equation of the form 
\[
p_{\beta\alpha}+p_{\beta\beta}=K_{2}
\]

These two final equations make the parameters $p_{\beta\alpha}$ and
$p_{\beta\beta}$ identifiable.

Thus all the parameters for $\mathbf{M}$ are identifiable. Hence
for any matrix $\mathbf{M}\in\mathbb{C}_{3}$, we have $\mathbf{M}\in\mathcal{F}$.
Thus $\mathbb{C}_{3}\subseteq\mathcal{I}$.
\end{proof}
\bibliographystyle{amsplain}
\nocite{*}
\bibliography{research2}

\providecommand{\bysame}{\leavevmode\hbox to3em{\hrulefill}\thinspace}
\providecommand{\MR}{\relax\ifhmode\unskip\space\fi MR }
\providecommand{\MRhref}[2]{%
  \href{http://www.ams.org/mathscinet-getitem?mr=#1}{#2}
}
\providecommand{\href}[2]{#2}
\begin{thebibliography}{10}

\bibitem{Anderson1957}
T.W. Anderson and Leo~A. Goodman, \emph{Statistical inference about markov
  chains}, Annals of Mathematical Statistics \textbf{28} (1957), 89--109.

\bibitem{Avery1999}
Peter~J. Avery and Daniel~A. Henderson, \emph{Fitting markov chain models to
  discrete state series such as dna sequences}, Applied Statistics \textbf{48}
  (1999), 53--61.

\bibitem{Bartlett1951}
M.S. Bartlett, \emph{The frequency goodness of fit test for probability
  chains}, Mathematical Proceedings of the Cambridge Philosophical Society
  \textbf{47} (1951), 86--95.

\bibitem{Billingsley1961}
Patrick Billingsley, \emph{Statistical methods in markov chains}, The Annals of
  Mathematical Statistics \textbf{32(1)} (1961), 12--40.

\bibitem{E.L.Lehman2003}
G.~Casella and E.L. Lehman, \emph{Theory of point estimation}, Springer, 2003.

\bibitem{Craig2002}
Bruce~A. Craig and Peter~P. Sendi, \emph{Estimation of the transition matrix of
  a discrete-time markov chain}, Health Economics \textbf{11} (2002), 33--42.

\bibitem{A.P.Dempster1977}
AP~Dempster, NM~Laird, and DB~Rubin, \emph{Maximum likelihood from incomplete
  data via the {EM} algorithm}, Journal of the Royal Statistical Society, Ser.
  B \textbf{39} (1977), 1--38.

\bibitem{Doob1953}
J.~L. Doob, \emph{Stochastic processes}, John Wiley and Sons, New York, 1953.

\bibitem{Ghosh2017}
Atanu~Kumar Ghosh and Arnab Chakraborty, \emph{Use of em algorithm for data
  reduction under sparsity assumption}, Computational Statistics \textbf{32}
  (2017), no.~2, 387--407.

\bibitem{D.K.Kim1995}
D~K Kim and J~M~G Taylor, \emph{The restricted {EM} algorithm for maximum
  likelihood estimation under linear restrictions on the parameters}, Journal
  of the American Statistical Association \textbf{90} (1995), 708--716.

\bibitem{Little1987}
R.~J.~A. Little and D.~B. Rubin, \emph{Statistical analysis with missing data},
  John Wiley, 1987.

\bibitem{A.Louis1982}
T.~A. Louis, \emph{Finding the observed information matrix when using the {EM}
  algorithm}, Journal of the Royal Statistical Society, Ser. B \textbf{44}
  (1982), 226--233.

\bibitem{G.J.McLachlan2008}
G.~J. McLachlan and T.~Krishnan, \emph{The em algorithm and extensions}, John
  Wiley, 2008.

\bibitem{X.Meng1991}
X.~Meng and D.~B. Rubin, \emph{Using {EM} to obtain asymptotic
  variance-covariance matrices: The {SEM} algorithm}, Journal of the American
  Statistical Association \textbf{86} (1991), 899--909.

\bibitem{Meng1991}
Xiao-Li Meng and Donald~B. Rubin, \emph{Using {EM} to obtain asymptotic
  variance-covariance matrices: The {SEM} algorithm}, Journal of the American
  Statistical Association \textbf{86(416)} (1991), 899--909.

\bibitem{N.Z.Shi2005}
N.~Z. Shi, S.~R. Zheng, and J.~Guo, \emph{The restricted {EM} algorithm under
  inequality restrictions on the parameters}, Journal of Multivariate Analysis
  \textbf{92} (2005), 53--76.

\bibitem{M.Tan2003}
M.~Tan, G.L. Tian, and H.B. Fang, \emph{Estimating restricted normal means
  using the em-type algorithms and ibf sampling}, World Scientific, New Jersey,
  2003.

\bibitem{G.L.Tian2008}
G.~L. Tian, K.~W. Ng, and M.~Tan, \emph{{EM}-type algorithms for computing
  restricted {MLE}s in multivariate normal distributions and multivariate
  t-distributions}, Computational Statistics and Data Analysis \textbf{52}
  (2008), 4768--4778.

\end{thebibliography}

\end{document}